\documentclass[letterpaper, 10 pt, conference]{ieeeconf}
\usepackage{IRM}
\begin{document}

%\title{\LARGE \bf Convex Optimization and Learning of Contracting\\  Nonlinear State Estimators}

\title{\LARGE \bf Contracting  Nonlinear Observers:  Convex Optimization and\\ Learning from Data}

%\title{\LARGE \bf Convex Sets of Contracting Nonlinear Observers and\\Learning of State Estimators}

%\title{\LARGE \bf Unifying Contracting Nonlinear Observer Design and\\ Prediction Error Identification}

\renewcommand{\r}{\rho}
\renewcommand{\D}{\Delta}
\newcommand{\z}{\hat x}

\IEEEoverridecommandlockouts
\author{
Ian R. Manchester$^{1}$ 
\thanks{*This work was supported by the Australian Research Council}
\thanks{$^{1}$ I. Manchester is with the
Australian Centre for Field Robotics, School of Aerospace, Mechanical and Mechatronic Engineering,, University of Sydney, Australia. 
{\tt\small ian.manchester@sydney.edu.au}}}
\maketitle

%\documentclass{ifacconf}
%
%\usepackage{graphicx}      % include this line if your document contains figures
%\usepackage{natbib}        % required for bibliography
%\usepackage{IRM_nothm}
%===============================================================================
%\begin{document}
%\begin{frontmatter}
%
%\title{Convex Sets of Contracting Nonlinear Observers\thanksref{footnoteinfo}} 
%% Title, preferably not more than 10 words.
%
%\thanks[footnoteinfo]{Sponsor and financial support acknowledgment
%goes here. Paper titles should be written in uppercase and lowercase
%letters, not all uppercase.}
%
%\author[First]{Ian R. Manchester} 
%
%
%\address[First]{ACFR (e-mail:).}
%
\begin{abstract}
A new approach to design of nonlinear observers (state estimators) is proposed. The main idea is to (i) construct a convex set of dynamical systems which are contracting observers for a particular system, and (ii) optimize over this set for one which minimizes a bound on state-estimation error on a simulated noisy data set. We construct convex sets of continuous-time and discrete-time observers, as well as contracting sampled-data observers for continuous-time systems. Convex bounds for learning are constructed using Lagrangian relaxation. The utility of the proposed methods are verified using numerical simulation.
\end{abstract}
%
%\begin{keyword}
%Five to ten keywords, preferably chosen from the IFAC keyword list.
%\end{keyword}
%
%\end{frontmatter}
%

\section{Introduction}
The problem of estimating the ``internal state'' (a.k.a. hidden or latent variables) of a dynamical system is one of the canonical problems in control engineering, and similar problems are encountered in time-series prediction, machine learning, signal processing, and may other fields. For linear systems a comprehensive methodology is available that combines computational simplicity, stability, and certain types of statistical or deterministic optimality (e.g. \cite{anderson2012optimal, petersen1999robust}). 

For non-linear systems, on the other hand, compromises are necessary. In all but a few isolated cases there is no finite-dimensional representation of the statistically-optimal estimator, so one must sacrifice one or more of optimality, global stability, or computational simplicity. 

Extended Kalman filters retain the simple structure of the linear case, but are suboptimal and stability is generally difficult to establish \cite{anderson2012optimal}. More accurate but computationally intensive approaches include particle filters a.k.a. sequential Monte Carlo methods \cite{doucet2001sequential}, Moving-horizon estimation \cite{rao2003constrained}, and set-valued state estimators \cite{shamma_approximate_1997}. Notably, particle filters and set-valued estimators maintain high-dimensional representations of a distribution (or set) of possible states, rather than a simple point estimate, while moving-horizon estimation maintains a long history of recorded inputs and outputs. Therefore even for low order systems, the internal state of the estimator is of  high dimension.

The term ``nonlinear observer'' usually refers to a state estimator that abandons any claim to statistical optimality, but  is a dynamical system of low dimension (equal or similar dimension to the true system), which is simple to implement, and for which global or semi-global stability can be established analytically or computationally \cite{praly2015observers}. %They are often analyzed in a deterministic setting.

 %While not entirely uniform in the literature, it is common to use the term {\em nonlinear observer design} in the context of a deterministic model with an emphasis on proving global or semi-global convergence to a perfect estimate, and {\em nonlinear state estimation} in the context of an stochastic or uncertain model with an emphasis on approximating the statistically optimal estimator. Such an estimator typically gives more information than a simple point estimate of the state, e.g. a posterior distribution of possible states or a set of possible states.

%Observers focus on computational simplicity and stability at the expense of optimality. Generally only provide a point estimate.

Over the years, many methods of observer design have been developed, including those based on geometric methods \cite{krener_linearization_1983}, high-gain designs \cite{khalil_high-gain_2014}, immersion and invariance \cite{karagiannis_invariant_2008}, and finding a nonlinear transformation (possibly with excessive coordinates) to a stable linear observer \cite{andrieu_existence_2006}.

Computational approaches based on convex optimization (in particular, linear matrix inequalities and semidefinite programming) have been used to construct observers for  systems with nonlinearities characterized by monotonicity  \cite{arcak_nonlinear_2001, fan_observer_2003} and $L^2$ gain (e.g. \cite{coutinho_robust_2009} and references therein). Most approaches are based on the search for a quadratic Lyapunov function, while in \cite{ebenbauer_polynomial_2005} sum-of-squares programming was used to find observers for polynomial systems via certain non-quadratic Lyapunov functions. Sum-of-squares relaxes the search for sign-definite polynomials to a semidefinite program \cite{parrilo2003semidefinite}, and will also be applied in the present paper.

Several papers have addressed observer design using contraction analysis, e.g. \cite{Lohmiller98, sanfelice_convergence_2012, dani_observer_2015}. Indeed, observer design was one of the original motivations for contraction analysis as introduced in  \cite{Lohmiller98}.
Contraction analysis is based on the study of a nonlinear system by way of its differential dynamics (a.k.a. variational system) along solutions.  Roughly speaking, since the differential dynamics are linear time-varying (LTV), many techniques from linear systems theory can be directly applied. A central result is that if all solutions of a smooth nonlinear system are locally exponentially stable in a common metric, then all solutions are globally exponentially stable. %A Riemannian contraction metric can be thought of as a family of quadratic Lyapunov functions for the LTV differential dynamics, and the Riemannian distance is then an incremental Lyapunov function for the original nonlinear system. 
Historically, basic convergence results on contracting systems can be traced back to the results of \cite{lewis1949metric} in terms of Finsler metrics, further explored in \cite{forni2014differential}, while convex conditions for existence and robustness of limit cycles were given in \cite{manchester2013transverse}, and extensions to constructive feedback design were given in \cite{manchester_control_2017}. 

In this paper we will construct convex sets of contracting observers for continuous and discrete-time nonlinear systems, as well as sampled-data observers, i.e. discrete-time observers for continuous-time systems. Sampled-data observer design is complicated by the inability to exactly solve a nonlinear differential equation over a sampling interval, and simple numerical integration techniques can cause instability \cite{arcak_framework_2004}. We will construct sampled-data observers via trapezoidal approximation, and prove that they remain stable.

Having constructed a convex set of observers, we propose optimizing over this set to find an observer that minimizes a measure of state estimation error on a finite data set. That is, our approach is based on {\em learning from data}, and is similar to the prediction error methods in black-box nonlinear system identification \cite{Sjoberg95, LjungBook} and recurrent neural networks 
\cite{pascanu2013difficulty}, and snapshot-based model reduction \cite{lucia2004reduced}. The central difference is that in prediction-error method the objective is to find a dynamical model of the system, whereas in this paper we assume a dynamical model is known, and the objective is to optimize a state estimator. This could be considered a form of regularization of prediction error methods: the identified model must behave as an observer for a certain known system. %The approach of mixing hard constraints and sampled data is also similar in philosophy to certain approaches to robust control design, in which some essential properties are imposed strictly (e.g. robust stability), while others are imposed approximately via sampling (e.g. performance)  \cite{calafiore_control_2004}.

This paper builds upon a framework introduced for system identification and model reduction in  \cite{tobenkin2010convex, tobenkin_convex_2017}. In particular, we make use of similar convex parameterizations of stable systems and bounds on nonlinear system behaviour as those papers. The novel contributions of this paper are: (i) we show how these techniques can be applied to the observer design problem, (ii) we provide new convex sets of models for CT systems and exponentially contracting systems, (iii) we show that trapezoidal integration of CT observers preserves stability, and (iv) we propose a method to optimize mean-square error of observers based on simulated data. The benefits of the proposed approach are illustrated using simulations.

\section{Preliminaries}

\subsection{Notation}
In this paper, we will consider both continuous-time (CT) and discrete-time (DT) systems and signals. Let $\Z^+$ and $\R^+$ denote the non-negative integers and reals, respectively. Let $\ell_2^n$ denote the set of square-summable signals $\Z^+\to\R^n$. Similarly, $L_2^n$ denotes the subspace of square-integrable signals $\R^+\to \R^n$. We use subscripts for time indexing, i.e. for a signal $x$, $x_t\in R^n$ denotes its value at time $t \in \Z^+$ or $\R^+$. For a symmetric matrix, $A>0$ ($A\ge 0$) denotes positive-definiteness (semidefiniteness), and $A>B$ ($A\ge B$) means $A-B>0$ ($A-B\ge 0$). For a vector $x$, $|x|$ is the Euclidean norm, and $|x|^2_P$ is shorthand for $x'Px$ where $P>0$. We will also abuse notation slightly by using the shorthand $|H|^2_P =H'PH$ even if $H$ is matrix with more than one column.

\subsection{Observer Design}
In this subsection we describe the deterministic nonlinear observer design problem. In Section \ref{sec:learning} we will extend consideration to stochastic models with disturbances and measurement errors.

We consider a ``true system'' model with finite-dimensional state $x\in\R^n$,  driven by a known input $u\in\R^m$ and producing a measured output $y\in\R^p$. The dynamic and sensor models are given by the following equations
\begin{equation}\label{eq:sys}
\sigma x_t = a(x_t,u_t), \quad y = g(x_t,u_t),
\end{equation}
where the operator $\sigma$ represents $\frac{d}{dt}$ for CT models, and the shift operator $x_t\mapsto x_{t+1}$ for DT models. In the CT case, for simplicity of development we will assume existence and uniqueness of solutions for all $t\in[0,\infty)$.

In this paper, an observer for \eqref{eq:sys} is another dynamical system with a state estimate as its internal state $\z\in\R^n$, which takes the input and output of \eqref{eq:sys} as its own inputs:
\begin{equation}\label{eq:general_obs}
\sigma \z_t= \phi(\z_t,u_t,y_t),
\end{equation}
While more general notions of an observer are possible with different dimension to the true system \cite{praly2015observers}, for this paper we restrict attention to the simplest case when the observer's internal state is $\hat x$. 
We will assume that the functions $a, g, \phi$ are continuously differentiable in their arguments.

The observer design problem is to find a function $\phi$ in \eqref{eq:general_obs} such that $\hat x_t$ is uniquely-defined on $t\in[0,\infty)$ (which we call well-posedness), and guaranteed to converge (in some sense) to $x_t$ as $t\rightarrow\infty$. In this paper, we will focus on the following forms of convergence:
\begin{defn}
System \eqref{eq:general_obs} is said to be an {\em $L^2$ observer} for \eqref{eq:sys} if  for all $x_0\in\R^n, \z_0\in\R^n$
$
x_t-\hat x_t \in L_2^n
$
for CT or $\ell_2^n$ for DT observers, and if $\hat x_0=x_0$ then $\hat x_t=x_t$ for all $t\ge 0$. System \eqref{eq:general_obs} is said to be an {\em exponential observer} with rate $\lambda>0$ if the stronger condition holds that there exists a function $b:\R^n\times\R^n\to\R^+$  with $b(x,x)=0$ for all $x$, such that
 for all $x_0, \z_0\in\R^n$:
\begin{equation}\label{eq:exp}
|x_t-\hat x_t|\le e^{-\lambda t}b(x_0,\z_0).
\end{equation}
\end{defn}

\subsection{Contraction Analysis}

Contraction analysis \cite{Lohmiller98} studies the convergence of solutions of a dynamical system to {\em each other}, rather than to a particular pre-defined ``nominal'' solution as in standard Lyapunov analysis. Thus stability of a system can be established independent of knowledge of a particular solution, a property which is particularly well suited to observer design.

Contraction analysis is based on the study of the {\em differential dynamics} (a.k.a. variational, linearized, prolonged), defined along solutions of a system of the form \eqref{eq:general_obs}:
\[
\sigma \delta_t = \pder[\phi]{\z}\delta_t.
\]
A central result of \cite{Lohmiller98} for the CT case is that if there exists a uniformly bounded metric $M(x)$ such that
\[
\dot M + \pder[\phi]{x}'M+M\pder[\phi]{x}\le -2\lambda M,
\]
where $\dot M = \pder[M]{x}\phi$, 
then the system is contracting with rate $\lambda$. In the DT case we have the similar condition
\[
\pder[\phi]{x}'M(x_{t+1})\pder[\phi]{x}-M(x_t)\le (1-e^{-\lambda}) M(x_t).
\]

\subsection{Bounds and Identities for Quadratic Functions}
Throughout the paper we will frequently use the following simple property. 
Concave functions $g(c,P)=-c'P^{-1}c$ with $c\in\R^n, P=P'>0$ obey the upper bound:
\begin{equation}\label{eq:quadbound}
-c'P^{-1}c\le b'Pb-2b'c
\end{equation}
where the right-hand side is a convex function of $c, P$ for any fixed $b\in\R^n$. Inequality \eqref{eq:quadbound} 
follows directly from the expansion
$
 b'Pb-2b'c+c'P^{-1}c =   |c-Pb|^2_{P^{-1}} \ge 0.
$
From this expansion it is also clear that the bound  \eqref{eq:quadbound}  is tight if $c=Pb$. 

We will also frequently use the following version of the polarization identity, valid for $Q=Q'>0$, arbitrary matrices $E, F$ of appropriate dimension, and arbitrary scalar $h>0$, which follows from expansion of the right hand side:
\begin{equation}\label{eq:polarization}
  E'QF+F'QE = \tfrac{1}{h}\left[|E+\tfrac{h}{2}F|^2_Q-|E-\tfrac{h}{2}F|^2_Q\right].
\end{equation}

\section{Convex Sets of Nonlinear Observers}
In this paper, a {\em system set} is a pair $(\phi, \Theta)$ where $\phi:\RR^n\times\RR^m\times\RR^p \times\RR^q\rightarrow \RR^n$ is a continuously differentiable function, and $\Theta \subset \R^q$ is a set of $q$-dimensional parameter vectors. Associated with $(\phi, \Theta)$ are state-space dynamical systems of the form
\begin{eqnarray}
\sigma \z_t &=&\phi(\z_t,u_t, y_t,\theta),\label{exsys1}
\end{eqnarray}
where $\z\in\RR^n, u\in\RR^m, y\in\RR^p$, $\theta\in\Theta$. 

We define a  {\em contracting observer set} for a system \eqref{eq:sys} as any system set such that, for all $\theta\in\Theta$, the following two conditions are satisfied: 
\begin{enumerate}
\item \textbf{Contraction:} any pair trajectories $\z^a, \z^b$ of \eqref{eq:general_obs} with the same inputs $(u,y)$ but with different initial conditions $\z^a_0, \z^b_0$ converge, in one of the following two senses
\begin{enumerate}
  \item $L^2$ contraction, i.e. $
\z^a-\z^b \in L^2
$
 (or $\ell^2$ for DT observers).
\item exponential contraction with rate $\lambda>0$, i.e.
\begin{equation}\label{eq:exp}
|\z^a_0- \z^b_0|\le e^{-\lambda t}b(\z^a_0, \z^b_0).
\end{equation}
function $b:\R^n\times\R^n\to\R^+$  with $b(x,x)=0$.
\end{enumerate}
\item \textbf{Correctness:} when initialised with $\hat x_0=x_0$, the observer matches the true system, i.e. $\hat x_t = x_t$ for all $t\ge 0$. This is the case if and only if \begin{equation}\label{eq:correct}
  a(x,u) = \phi(x,u,g(x,u),\theta) \quad \forall x\in\R^n, u\in\R^m.
\end{equation}
%\item \textbf{Lipschitz} there is some $L>0$ such that for all $u\in\R^m, y\in\R^p, z^a\in\R^q, z^b\in\R^q$ the following holds:
%\[
%|\psi(z^a,u,y)-\psi(z^b,u,y)|\le L|z^a-z^b|
%\]
\end{enumerate}
In what follows, we will usually drop the $\theta$ argument, and speak directly of searching over functions $\phi(\z,u,y)$.

It is obvious that these two conditions are sufficient for \eqref{eq:general_obs} to be an observer for \eqref{eq:sys} for any $\theta\in\Theta$: correctness implies that the true state $x$ is a particular solution of the observer, while contraction implies that all solutions of the observer converge to each other, hence all solutions of the observer converge to the true state. This characterization of a contracting observer is a special case of the idea of {\em virtual system} used to study observers and more general synchronization behaviors in \cite{wang_partial_2004}.

%As discussed in Section \ref{sec:prelim}, these conditions are sufficient to ensure that $\z \rightarrow x$ as $t\rightarrow\infty$, in either an $L^2$ or exponential sense.

\subsection{Convex Sets of Contracting Continuous-Time Observers}\label{sec:CT}

In this section we will construct convex sets of contracting nonlinear observers for CT systems of the form
\begin{equation}\label{eq:ct_sys}
  \dot x_t = a(x_t,u_t), \quad y_t = g(x_t,u_t).
\end{equation}
We will construct our observers in the following {\em implicit} representation:
\begin{equation}\label{eq:ct_obs_ef}
  \ddt e(\z_t) = E(\z_t)\dot \z_t = f(\z_t,u_t,y_t).
\end{equation}
where $E(\z) = \pder{\z}e(\z)$. As noted in \cite{tobenkin2010convex, tobenkin_convex_2017} an implicit system representation significantly expands the flexibility of the system set while retaining convexity.

We will ensure that $E(\z)$ is invertible for all $\z$, and hence the observer can be rewritten in the explicit form
\begin{equation}\label{eq:CT_explicit}
\dot \z_t = \phi(\z_t,u_t,y_t)=  E(\z_t)^{-1}f(\z_t,u_t,y_t).  
\end{equation}
In principle, our results apply to search over infinite dimensional space of functions $e, f$ that are continuously differentiable, but in practical implementations $e$ and $f$ will be linearly parameterized by a finite-dimensional vector $\theta$:
\begin{align}
&e_\theta(\z) = \sum_{i = 0}^q \theta_i e_i(\z),\ 
f_\theta(\z,u,y) = \sum_{i = 0}^q \theta_i f_i(\z,u,y), 
\label{eq:linparam}
\end{align}
where $\theta\in\RR^q$ and each basis function $e_i:\ \RR^n\to\RR^n, f_i:\ \RR^n\times\RR^m\times\RR^p\to\RR^n$, is continuously differentiable in its arguments. In particular, if the basis functions are polynomials, then sum-of-squares programming can be used to guarantee global sign-definiteness \cite{parrilo2003semidefinite}.

To show contraction, we make use of the differential dynamics of \eqref{eq:ct_obs_ef}, defined along a particular solution $\z, u, y$:
\begin{equation}
  \ddt (E(\z_t)\d_t) = F(\z_t,u_t,y_t)\d_t.
\end{equation}
where $E(\z) = \pder{\z}e(\z)$ and  $F(\z,u_t,y_t) = \pder{\z}f(\z,u_t,y_t)$.
Following \cite{tobenkin2010convex}, we will consider contraction metrics of the form
$
V(\z,\d) = \d'E(\z)'P^{-1}E(\z)\d
$
where $P=P'>0$ is an auxiliary matrix variable. Contraction with respect to this metric can be established by differential dissipation inequalities:
\begin{equation}\label{eq:CT_contraction_H}
\ddt V(\z,\d) =	\d'(E'P^{-1}F+F'P^{-1}E)\d \le -\d' H\d
\end{equation}
for all $\d,\z,u,y$. In particular, $L^2$ contraction follows from any $H>0$, while exponential contraction with rate $\lambda$ follows from the choice $H = 2\lambda E'P^{-1}E$. The optimal bound $b(\z^a_0,\z^b_0)$ in \eqref{eq:exp} is the Riemannian energy with respect to $V$ \cite{manchester_control_2017}.

The correctness condition for observer \eqref{eq:ct_obs_ef} and system \eqref{eq:ct_sys} is the following linear (and hence convex) constraint:
\begin{equation}\label{eq:CT_correctness}
E(x)a(x,u) = f(x,g(x,u),u) \quad \forall x\in \R^n, u\in \R^m
\end{equation}
With $E$ invertible, this clearly implies \eqref{eq:correct} via \eqref{eq:CT_explicit}.

So the objective is to find an observer \eqref{eq:ct_obs_ef} with $E(\z)$ invertible for all $\z$ that satisfies the contraction condition \eqref{eq:CT_contraction_H} and the correctness condition \eqref{eq:CT_correctness}. The difficulty is that both invertibility of $E(\z)$ and \eqref{eq:CT_contraction_H} are nonconvex constraints on $\theta$.
%At this point, if we can construct a large set of dynamical systems \eqref{eq:ct_obs_ef} with $E(\z)$ invertible for all $\z$, which satisfy the contraction condition \eqref{eq:CT_contraction_H} and the correctness condition \ref{eq:CT_correctness}, then we are done. When searching over functions $e$ and $f$, it is clear that \eqref{eq:CT_correctness} is a linear (and hence convex) constraint. The main difficulties are that the contraction condition \eqref{eq:CT_contraction_H} is not jointly convex in $e, f, P$, and that the set of functions $e$ for which $E$ is invertible is non-convex.
We first recall the following useful result \cite[Thm. 1]{tobenkin_convex_2017}, giving a convex constraint for invertibility:
\begin{lem}\label{lem:wellposed}
	Let $E(x) = \pder[e]{x}$ for continuously differentiable function $x\mapsto e(x)$  mapping $\R^n \to \R^n$. Suppose
\begin{equation}\label{eq:well_posed_E}
 E(x) +E(x)'\ge \mu I, \quad \forall x\in\R^n
\end{equation}
 for some $\mu>0$, then $e$ is a bijection, and $E$ is non-singular for all $x$.
\end{lem}
%
%
%\begin{thm}
%Let $E(x) = \pder[e]{x}$ and $Q(x) = \pder[q]{x}$ for continuously differentiable functions $x\mapsto e(x)$ and $q\mapsto q(x)$, both mapping $\R^n \to \R^n$. Suppose
%\begin{equation}\label{eq:well_posed_QE}
% Q(x)E(x) +E(x)'Q(x)' \ge \mu I
%\end{equation}
% for some $\mu>0$,
%then both $q$ and $e$ are bijections, and $E$ and $Q$ are non-singular for all $x$.
%\end{thm}
%\begin{proof}
%  The result follows from Theorem 5 of \cite{tobenkin_convex_2017} applied to the function $q(e(x))$, and the fact that if this function is a bijection, then so are both $q$ and $e$, and hence their Jacobians are nonsingular.
%\end{proof}
%

%The utility of this is that for any fixed function $q(x)$, \eqref{eq:well_posed_QE} is a convex constraint on the set of functions $e$.

%To convexify the contraction constraint, we present two 

%In this section we construct convex relaxations for the non-convex constraints above.

%The invertibility of $E$ is simply imposed by the conservative but convex condition
%\begin{equation}\label{eq:well_posedness}
%E(\z)+E(\z)'\ge \mu I
%\end{equation}
%for some scalar $\mu>0$.

We now present two alternative choices of convex constraints guaranteeing the contraction condition \eqref{eq:CT_contraction_H}. 

\begin{enumerate}
	\item \textbf{CT1:} We restrict the space of functions $e$ to the class of $e(x) = Px$. Then \eqref{eq:CT_contraction_H} reduces to the condition\begin{equation}\label{eq:CT_contraction_P}
  F+F'+H\le 0
\end{equation}
with $H>0$ for $L^2$ contraction or $H=2\lambda P$ for exponential contraction with rate $\lambda$. This is jointly convex in $e, f, P$ and quasi-convex in $\lambda$.
\item \textbf{CT2:} Condition \eqref{eq:CT_contraction_H} is replaced with the stronger condition 
\begin{align}
	(E-\tfrac{h}{2}F)+(E-\tfrac{h}{2}F)'-P & \notag\\
	- (E+\tfrac{h}{2}F)'P^{-1}(E+\tfrac{h}{2}F) - H&\ge 0.\label{eq:CT_contraction_relax}
\end{align}
for some scalar $h>0$ and arbitrary matrix $H>0$ for $L^2$ contraction or $H=2\lambda E'PE$ for exponential contraction with rate $\lambda$.
%\item \textbf{CT3:} Condition \eqref{eq:CT_contraction_H} is replaced with the stronger condition, valid for any $Q(z)$
%\begin{align}
%	&Q'(E-\tfrac{h}{2}F)+(E-\tfrac{h}{2}F)'Q-Q'PQ \notag\\&- (E+\tfrac{h}{2}F)'P^{-1}(E+\tfrac{h}{2}F) -H\ge 0.\label{eq:CT_contraction_relax}
%\end{align}
%with $H>0$ for $L^2$ contraction or $H=2h\lambda E'PE$ for exponential contraction with rate $\lambda$.
\end{enumerate}

Both of these conditions are jointly-convex in $e,f, P$, and the exponential contraction conditions are quasi-convex in $\lambda>0$. Therefore the rate $\lambda$ can be maximized via bisection search with a convex feasibility problem at each step.

The following theorem summarizes the main results of this section:
\begin{thm}\label{thm:CTmain}
  Given a true system \eqref{eq:ct_sys}, a convex set of contracting observers is given by the system set \eqref{eq:ct_obs_ef}, \eqref{eq:linparam}, with $\Theta$ defined by the convex constraints \eqref{eq:well_posed_E}, \eqref{eq:CT_correctness}, and either \eqref{eq:CT_contraction_P} or \eqref{eq:CT_contraction_relax}, with $H$ as specified for $L^2$ or exponential contraction.
\end{thm}
\begin{proof}
As remarked above, well-posedness follows from \eqref{eq:well_posed_E} via Lemma \ref{lem:wellposed}. The correctness condition \eqref{eq:CT_correctness} is clearly equivalent to \eqref{eq:correct} via the explicit representation \eqref{eq:CT_explicit}.

To prove that \eqref{eq:CT_contraction_P} implies contraction for CT1, note that with $e(\z)=P\z$ the observer is 
\[
\dot \z = P^{-1}f(\z,u,y)
\] 
and contraction \eqref{eq:CT_contraction_H} follows directly from \eqref{eq:CT_contraction_P} since $E=P$.

To show that \eqref{eq:CT_contraction_relax} implies contraction for CT2, we first apply \eqref{eq:polarization} to \eqref{eq:CT_contraction_H}
	gives the equivalent condition
\[
	\d'\left[|E+\tfrac{h}{2}F|^2_{P^{-1}}-|E-\tfrac{h}{2}F|^2_{P^{-1}}\right]\d\le -2h\lambda \d E'P^{-1}E\d
\]
Then applying \eqref{eq:quadbound} with $c = (E-\tfrac{h}{2}F)\d$ and $b =\d$ shows that \eqref{eq:CT_contraction_relax} guarantees \eqref{eq:CT_contraction_H}.
\end{proof}

\begin{remark}
	By using different quantities for $b$ in \eqref{eq:quadbound} one obtains a family of convex sets guaranteeing contraction. For example, an iterative refinement procedure could start with CT1 or CT2 and take advantage of the fact that \eqref{eq:quadbound} is tight when $c=Pb$ and set $b_{k+1} = P_k^{-1}c_k = P_k^{-1}(E_k-F_k)$, where $k$ indexes iteration number, such that the convex bound is tight at the result of the previous iteration. 
\end{remark}

%CT1 is similar to papers that use a quadratic Lyapunov function, e.g. \cite{arcak_framework_2004, fan_observer_2003, coutinho_robust_2009}

%CT2 is similar to \cite{tobenkin2010convex}. It is derived from the relation 

%Tight if $Q = P^{-1}(E-F)$. Start with CT1, refine
%The case $Q=I$ was considered in \cite{tobenkin2010convex}.

%Degree bounds on $Q$ for global satisfiability.

%To summarize?
%
%\begin{align}\label{eq:stab_LMIC}
%\begin{bmatrix}(E-F)+(E-F)'-P &(E+F)' & E'\\
%    (E+F) & P & 0\\
%    E & 0 & \frac{1}{4\lambda}P\\
%\end{bmatrix}&\ge 0.
%\end{align}
%where $P>0$.

\subsection{Convex Sets of Contracting Discrete-Time Observers}

This section contains parallel results for DT systems:
\begin{equation}\label{eqn:sys_DT}
x_{t+1}=a(x_t,u_t,), \quad y_t = g(x_t,u_t).
\end{equation}
The observers we will consider are of the implicit form
\begin{eqnarray}
e_\theta(\z_{t+1})&=&f_\theta(\z_t,u_t,y_t),\label{eq:DT_observer}
\end{eqnarray}
with $e,f$ parameterized as in \eqref{eq:linparam}, 
but we will ensure that $e$ is a bijection, and hence the observer can be written in the explicit form
\begin{equation}
  	\z_{t+1} =\phi(\z_t,u_t,y_t) =  e_\theta^{-1}(f_\theta(\z_t,u_t,y_t))
\end{equation}
Well-posedness, i.e. the fact that $e$ is a bijection, will be guaranteed by ensuring $E+E'\ge \mu I$, and applying Lemma \ref{lem:wellposed}. A solution can be generally be obtained rapidly by Newton's method.

The correctness condition in DT is
\begin{equation}\label{eq:DT_correctness}
  e(a(x,u))=f(x,u,g(x,u)) \quad \forall x\in\R^n, u\in\R^m.
\end{equation}
This constraint is linear (and hence convex) in $e$ and $f$ for any known functions $a, g$.

%As in the CT case the contraction condition is not jointly convex in $e,f,P$. %Again, as convex bound is constructed from \eqref{eq:quadbound}.

To study contraction, we make use of the differential dynamics of \eqref{eq:DT_observer}, given by
\begin{equation}
  E_{t+1}\d_{t+1} = F_t\d_t
\end{equation}
where we have used the shorthand 
$
E_{t}:=E(\z_{t})$ and $
F_{t}:=F(\z_{t},u_{t},y_{t})$.

Similarly to the CT case, contraction for the DT system will be evaluated via a metric of the form
$
V(\z, \d) = \d'E(\z)'P^{-1}E(\z)\d.
$
The contraction condition is then 
\begin{equation}\label{eq:DT_contraction}
  V_{t+1}-V_t = \d_t'(F'_tP^{-1}F_t-E_t'P^{-1}E_t)\d_t\le - \d_t' H\d_t
\end{equation}
again using shorthand $
V_{t}:=V(\z_{t},\d_t)$, and $H>0$ for $\ell^2$ contraction, and $H = (1-e^{-\lambda})E_t'P^{-1}E_t$ for exponential convergence with rate $\lambda$. Condition \eqref{eq:DT_contraction} is not jointly convex in $e,f,P$ due to the concave term $-E_t'P^{-1}E_t$.

Using \eqref{eq:quadbound} to convexify \eqref{eq:DT_contraction} with $b=\d_t$, $C=E_t\d_t$ we obtain the following:
\begin{align}
  E(\z)+E(\z)'-P-F(\z,u,y)'P^{-1}F(\z,u,y)-H\ge 0\label{eq:DT_relax}
\end{align}
for all $\z\in\R^n, u\in\R^m, y\in\R^p$. To summarize the construction in this section:
\begin{thm}
Given a true system \eqref{eqn:sys_DT}, a convex set of contracting observers is given by the parameterization \eqref{eq:DT_observer}, \eqref{eq:linparam} where $\Theta$ is defined by the following constraints
\begin{enumerate}
  \item Contraction: for some $P>0$, \eqref{eq:DT_relax} holds 
with $H>0$ arbitrary for $\ell_2$ contraction, and $H = 2\lambda E'PE$ for exponential contraction with rate $\lambda$.
\item Correctness: \eqref{eq:DT_correctness} holds for all $x,u$.
\end{enumerate}which are convex in the unknowns $e, f, P$.
\end{thm}
Note that well-posedness follows automatically from \eqref{eq:DT_relax} and Lemma \ref{lem:wellposed}.

\subsection{Contracting Sampled-Data Observers for CT Systems}\label{sec:SD}

It is common in applications that a dynamical model is known in continuous time (e.g. from physical laws), but an observer is to be implemented in discrete time on a computer.

For nonlinear systems, the challenge is that exact solution of CT differential equation is not available, so precise sampling is generally not possible. Simple methods, such as Euler integration, can destroy stability properties of CT observers \cite{arcak_framework_2004}. In this section we show that trapezoidal integration:
\[
z_{t+1}-z_t \approx \frac{h}{2}(\dot z_{t+1}-\dot z_t)
\]
where $h>0$ is the time-step, is particularly well-suited to sampled-data implementation of contracting observers.
It has long been known that trapezoidal integration enjoys favourable stability properties. In fact it is the most accurate linear multi-step method that guarantees stability when integrating arbitrary linear systems \cite{dahlquist_special_1963}.

Given a CT observer satisfying the conditions of Section \ref{sec:CT}, a sampled-data observer can be constructed as
\begin{equation}\label{eq:SD_observer}
  e(\z_{t+h})-\frac{h}{2}f(\z_{t+h},u_{t+h},y_{t+h}) = e(\z_{t})+\frac{h}{2}f(\z_{t},u_{t},y_{t}).
\end{equation}
Note that the update from $t\rightarrow t+h$ depends on measured data of $u$ and $y$ at both times $t$ and $t+h$.

\begin{thm}
	For a CT system \eqref{eq:ct_sys} and any CT observer satisfying the conditions of Theorem \ref{thm:CTmain}, the sampled-data observer given in \eqref{eq:SD_observer} is
well-posed,
 $\ell^2$ stable, and  ``correct'' in the sense that it matches the trapezoidal integration of the CT system
\end{thm}
\begin{proof}
We sketch the proof due to space restrictions. Well-posedness of the observer \eqref{eq:SD_observer} is guaranteed by Lemma \ref{lem:wellposed} if $E(\z)-\tfrac{h}{2}F(\z,u,y)+E(\z)'-\tfrac{h}{2}F(\z,u,y)'\ge \mu I$. If CT1 is used, then by construction $E+E'>0$ and $F+F'\le 0$ so clearly the observer is well-posed for any $h$. For CT2, by inspection the only term on the left-hand-side of \eqref{eq:CT_contraction_relax} that is not negative-semidefinite is $(E-\tfrac{h}{2}F)+(E-\tfrac{h}{2}F)'$. Therefore since $P>0$ this term is positive-definite, which is precisely the condition required for well-posedness.

Both CT1 and CT2 imply \eqref{eq:CT_contraction_H}. For the sampled data observer consider the 
differential dynamics:
\[
(E_{t+1}-\tfrac{h}{2}F_{t+h})\d_{t+h} = (E_t+\tfrac{h}{2}F_t)\d_t
\]
and metric 
\[
V(\z,\d,t) = \delta'(E_t-\tfrac{h}{2}F_t)'P^{-1}(E_t-\tfrac{h}{2}F_t)\delta
\]
Then $
V(\z_{t+h},\d_{t+h},t+h)-V(\z,\d,t)=\d(|E_t-\tfrac{h}{2}F_t|_{P^{-1}}^2-|E_t+\tfrac{h}{2}F_t|_{P^{-1}}^2)
$. By the polarization identity and \eqref{eq:CT_contraction_H} this difference is uniformly negative-definite, and hence the observer is $\ell^2$ contracting
\end{proof}

Note that exponential convergence with the same rate $\lambda$ is {\em not} guaranteed, since the metric  used to evaluate contraction is different.

\subsection{Flexibility of the Observer Sets}

%Converse results depend on the ``richness'' of the sets.
%
%
%Summary
%\begin{enumerate}
%	\item All contain all linear systems
%	\item CT1 contains systems not in CT2, and vice-versa
%	\item CT1 contains all quadratically stable continous-time systems
%	\item DT2 contains all explicit-form quadratically stable DT systems
%	\item CT2 contains all CT systems which, when approximated by trapezoidal, are quadratically contracting. (Dahlquist?)
%	\item DT1 contains all systems that are trapezoidal approximations of quadratically-contracting CT systems.
%\end{enumerate}
%What is the difference between 4 and 5? 5 not really true, DT2 requires explicit form.

%We provide convex sufficient conditions for nonlinear observer design. In general we do not expect them to be necessary. However, for the case of linear systems we show that they are necessary and sufficient. Although simpler convex formulations exist for linear systems, we note that some results for nonlinear observer design require the stronger assumption of {\em observability} (e.g. Assumption 2 in \cite{kazantzis1998nonlinear}). 

To be able to find observers for a wide class of systems, it is of course beneficial to have as flexible a set of observers as possible. The main result regarding flexibility is the following:

\begin{thm} Considering observers in the span of \eqref{eq:linparam}, \begin{enumerate}
  \item CT1, CT2, and DT contain all observers that can be written in the form
    \[
    \sigma\z_t = A_o\z_t + q(u_t,y_t)
    \]
    where $A_o$ is a Hurwitz stable matrix for CT or Schur stable matrix for DT.
    \item CT1 and DT additionally contain all nonlinear observers that are contracting with respect to a constant metric $V(\z,\d) = \d'M\d$ where $M=M'>0$.
\end{enumerate}
\end{thm}
we omit the proof due to space restrictions. Similar results are proved in \cite{tobenkin2010convex,tobenkin_convex_2017}.

  This implies that observers exist in these sets for any system of the form
  \[
  \sigma x_t = Ax_t + q(u_t,y_t),\quad y_t = Cx_t+r(u_t)
  \]
  where the pair $(A, C)$ is detectable, including all full-order observers for linear detectable systems.

From Theorem 1, it may appear that CT1 is more flexible than CT2. In fact CT1 contains systems that CT2 does not, and CT2 contains systems that CT1 does not. In particular, considering polynomial bases for $e,f$, CT1 contains the system:
\[
\dot\z = -\z-\z^3
\]
with $e(\z) = \z, f(\z) = -\z-\z^3$, since $F(\z) =-1-3\z^2$ is uniformly negative-definite.

CT2 does not contain this system:
It is clear from \eqref{eq:CT_contraction_relax} that $(E-F)+(E-F)'$ must be larger in magnitude than $(E+F)'P^{-1}(E+F)$, which is impossible if $E$ is constant and $F$ is unbounded in magnitude as $\z\mapsto\infty$, as is the case if $f$ is a polynomial of degree $>1$.

On the other hand, CT2 contains systems that CT1 does not. In particular, the system
\[
\dot\z =\phi(\z)= -\frac{\z+\z^3}{0.1+3\z^2}
\]
from $e(\z) = 0.1\z+\z^3, f(\z) = -\z-\z^3$. It can be seen from the graph of $\phi$ in Fig. \ref{fig:nonquad} that the slope is positive, and hence there is no constant $P>0$ such that $P\pder[\phi]{\z} +\pder[\phi]{\z}'P<0$. However, by direct evaluation this system satisfies \eqref{eq:CT_contraction_H} with $P=1, H=0.1$. However, it is contracting with respect to the metric $V(\z,\d) = \d'E(\z)'P^{-1}E(\z)\d = (0.1+3\z^2)^2\d^2$.

\begin{figure}\centering
  \includegraphics[width=0.85\columnwidth]{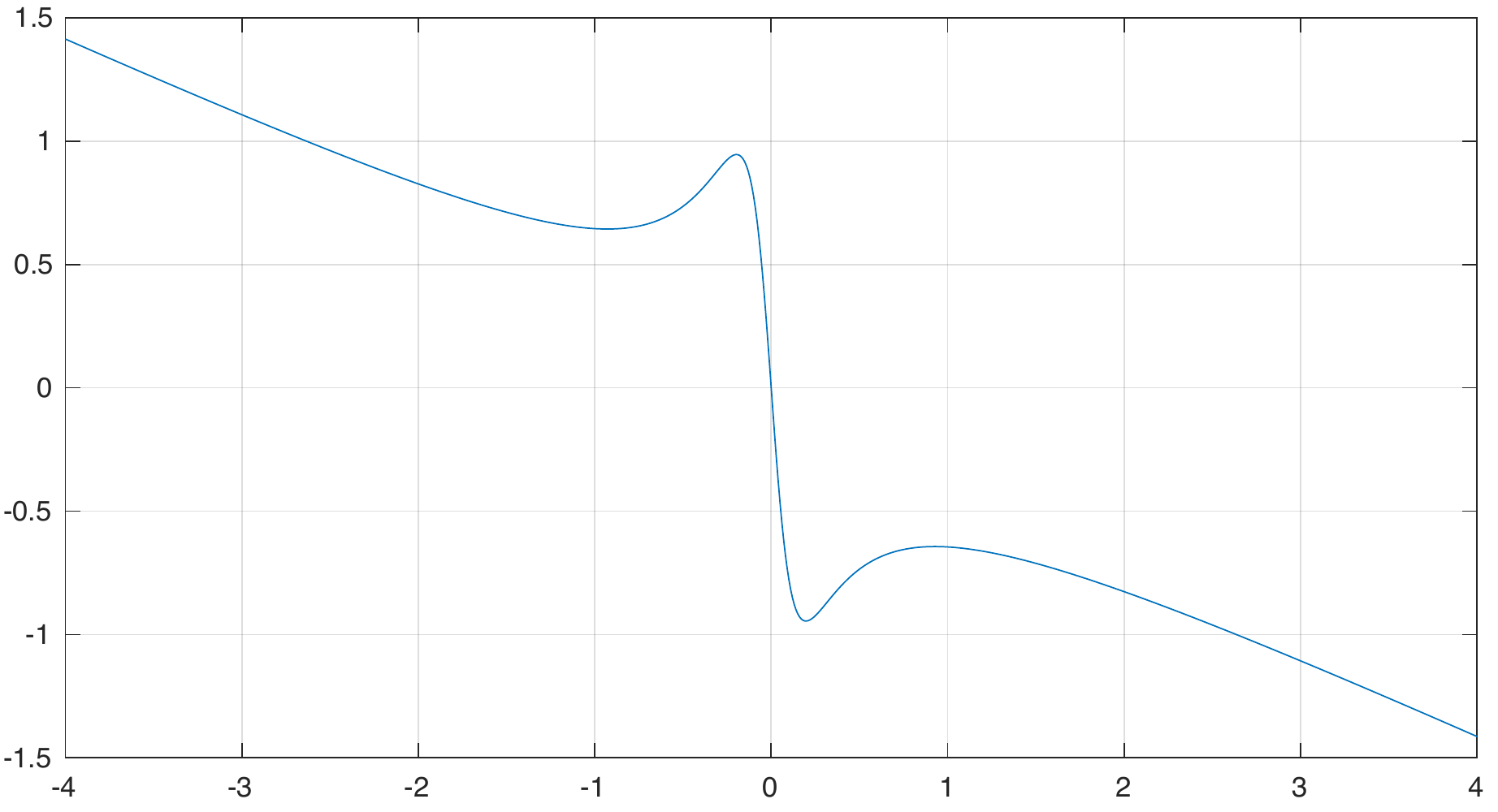}
  \caption{Graph of the function $\phi(\z) = -{(\z+\z^3)}/{(0.1+3\z^2)}$. The positive slope around $\z = \pm \tfrac{1}{2}$ implies that the system $\dot\z = \phi(\z)$ is not contracting with respect to a constant metric. However, it is in the set CT2.}
  \label{fig:nonquad}
\end{figure}

\section{Learning from Data}\label{sec:learning}
Within a set of stable observers, it is natural to try to find the one that is ``best'' in some sense. 
A standard way to formulate the problem is to specify a particular stochastic model of the true system, e.g. in the DT case
\begin{equation}\label{eq:DT_sys_noisy}
  x_{t+1} = a(x_t,u_t,w_t), \quad y = g(x_t,u_t,w_t)
\end{equation}
where $u_t$ is a known input, and $w_t$ represents measurement and/or process noise, and is unknown but drawn from a known probability distribution (we assume $w=0$ corresponds to a ``nominal'' model). Initial conditions may be known, or also drawn from a probability distribution. Then a measure of state estimation error is chosen, e.g. the mean-square error:
\[
J = \lim_{T\rightarrow\infty} \frac{1}{T}\sum_{t=0}^T\mathcal E[(x_t-\hat x_t)'(x_t-\hat x_t)],
\]
where $\mathcal E$ denotes the expectation operator (see, e.g., \cite{anderson2012optimal}). Unfortunately there is no computationally tractable solution to this problem for most nonlinear systems, so some form of approximation is required.

We propose an approach based on ``learning from data''. %This is similar to system identification and machine learning, though it is different in that we are searching over sets systems that are all observers for a known ``nominal'' system, and the data are generated from simulation of a known stochastic system model.
The procedure we propose is outlined as follows:
\begin{enumerate}
  \item Construct a ``flexible'' set of contracting nonlinear observers for the nominal ($w=0$) model.
  \item Draw one or more realisations of $w$ and simulate the stochastic model \eqref{eq:DT_sys_noisy}, collecting data sets $\tilde z = \{\tilde x_t, \tilde u_t, \tilde y_t | t = 0, 1, ..., T\}$ as the realisations of $x,u,y$.
  \item Optimize over the set of observers for one which minimizes the empirical mean-square error:
  \[
  J = \frac{1}{T}\sum_{t=0}^T|\hat x_t-\tilde x_t|^2
  \]
  where $\hat x$ are solutions of the observer simulated with data $\tilde u_t, \tilde y_t$, and an estimate of $\hat x_0$. When multiple realisations of $w$ are sampled, $J$ is the sum of the corresponding mean-square errors.
\end{enumerate}
If a method is available for simulating a stochastic CT model, then the same approach can be applied to optimize sampled-data observers by subsampling at fixed rate.

\begin{remark}
	In the linear case, the constraint that the learned observer is an observer for the nominal system means the resulting estimator is unbiased when $w$ is zero mean. In the nonlinear case this is not generally true, but it can be considered as a form of regularization that introduces a bias to ensure reliable behaviour on data  not in the training set.
\end{remark}

While the above recipe can be followed using generic nonlinear programming methods to minimize $J$, this may prove challenging since $J$ is a highly non-convex function of the observer parameters even in the linear case. We construct a convex approximation based on Lagrangian relaxation of linearized estimation error, similar to the method proposed in \cite{tobenkin_convex_2017} for system identification.

First we construct the linearization along the sampled ``true solution'' $\tilde x$, by introducing the local deviation $\tilde \D_t \approx \hat x_t-\tilde x_t$ obeying the dynamics
%ameg removed ":"
\begin{equation}
  E(\tilde x_{t+1}) \D_{t+1} = F(\tilde x_t,\tilde u_t) \tilde\D_t +
  \e_t.\label{eq:linsimdyn}
\end{equation}
with $\tilde\D_0=0$, where the equation errors $\e_t$ are defined as
\[
\e_t := e(\tilde x_{t+1}, u_{t+1}, y_{t+1})-f(\tilde x_t, \tilde u_t, \tilde y_t)
\]
Then the local mean-square error is
\begin{equation}
  J^0(\theta, \tilde z) = \frac{1}{T}\sum_{t=0}^T |\D_t|^2,
\end{equation}
Roughly speaking, when the state estimation error is small, i.e. $\tilde \D_t \approx \hat x_t-\tilde x_t$ then $J\approx J^0$. A more formal derivation for system identification problems can be found in \cite[Sec. V.B]{tobenkin_convex_2017}.
Note that if the observer dynamics are affine in $\z$, then $J$ and $J^0$ are identical.

%\subsection{Convex Upper Bounds for Linearized Simulation Error}

The problem of minimizing $J^0$ remains nonconvex, but a semidefinite programming bound can be constructed using Lagrangian relaxation \cite{tobenkin_convex_2017}. This is most clearly expressed be rewriting \eqref{eq:linsimdyn} in a ``lifted'' representation
\[
H(\theta, \tilde z)\vec\D = \vec\e(\theta, \tilde z)
\]
with the stacked vectors $\vec\D:=[\tilde\D_1', \tilde\D_2', ..., \tilde\D_T']'$ and $\vec\e(\theta, \tilde z): = [\e_0', \e_1', ... \e_T']'$ and $H(\theta, \tilde z)$ is a block lower-bidiagonal matrix that is affine in $\theta$ and is straightforward to construct from \eqref{eq:linsimdyn}. 

Then a convex upper bound for $J^0$ is given by
\begin{align*}
\hat J^0(\theta, \tilde z):= \sup_{\vec\D\in l_T^{nT}}&\{ |\vec\D|^2 -2\vec\D'(H(\theta, \tilde z)\vec\D-\vec\e(\theta, \tilde z)) \}.
\end{align*}
It can be shown that for all systems satisfying the contraction constraint, the supremum over $\vec\D$ is finite, and can be written in explicit form:
\begin{equation}\label{eq:Jhat_exp}
  \hat J^{0}(\theta, \tilde z) = \vec\e(\theta, \tilde z)'\mathcal H(\theta,z)^{-1}\vec\e(\theta, \tilde z)
\end{equation}
where $\mathcal H = H+H'-I$ or, via Schur complement, in the equivalent semidefinite programming representation:
\begin{equation}\label{eq:Jhat_sdp}
\hat J^{0}(\theta, \tilde z) =\min s\quad \textrm{ s.t. }  \bm{s & \vec\e(\theta, \tilde z)'\\
  \vec\e(\theta, \tilde z) &\mathcal H(\theta,z)}\ge 0.
\end{equation}

%\begin{equation}
%  \mathcal H= \bm{E_1+E_1'-I & -F_1' & 0 &  \cdots \\
%  		-F_1 & E_2+E_2'-I &-F_2' &  \ddots\\
%  		0 & -F_2 & \ddots &  \ddots\\
%  		\vdots & \ddots & \ddots &  \ddots}
%\end{equation}

Algorithms have been developed for efficient solution of problems with this structure \cite{Umenberger16, andersen_implementation_2010}. Both have computational complexity that is linear in the data length $T$, rather than cubic for a generic semidefinite programming solver.

The main result regarding $\hat J^0$ is the following:
\begin{thm}\label{thm:rieldt}
For all  $P=P'>0$ and signal data $\tilde z$, for every $\theta\in\Theta$,
$ J^0(\theta, \tilde z) \le \hat J^0(\theta, \tilde z)<\infty.
$
Furthermore, if the data set $\tilde z$ is generated with $w=0$, then the optimal values satisfy
$
J^{0,\star}(\tilde z) = \hat J^{0,\star}(\tilde z) =0.
$

%$\hat J^0$ is finite. Upper bound for $J^0$.
%On noise-free data it is zero.

\end{thm}
We omit the details of the proof because of space restrictions, since it is similar to \cite[Theorem 6]{tobenkin_convex_2017}.

\section{Examples and Discussion}
In this section we illustrate the method with some numerical simulations. The first example system we consider was previously studied in \cite{fan_observer_2003}:
\[
\dot x = a(x)=
\begin{bmatrix}
	x_2\\
    x_2-\frac{1}{3}x_2^3-x_2x_3^2\\
    x_2-x_3-\frac{1}{3}x_3^3-x_3x_2^2
\end{bmatrix}, \quad
y = g(x) = x_1
\]
The solution given in \cite{fan_observer_2003} relied on a certain ingenuity on the part of the designer, recognising the system can be decomposed into a linear part and a monotonic nonlinearity. In comparison, our method is quite ``plug and play'': we apply Theorem \ref{thm:CTmain} using CT1 as a contracting set and $f$ the set of degree-three polynomials in $\z, y$, the contraction constraint is imposed via sum-of-squares \cite{parrilo2003semidefinite}, and an observer is found in 1.41 seconds on 3.1 GHz Intel Core i7 with 16GB RAM. 

A bisection search over $\lambda$ reveals that $\lambda=1$ is the best rate of exponential convergence in this observer set, and with  $\lambda=1$ we computed the observer that minimized the $\ell^1$ norm of the coefficients of $f$, to encourage sparsity of coefficients. All computations were done using Yalmip \cite{YALMIP} and Mosek.

Fig. \ref{fig:ds_10} shows the evolution of the three states (solid lines) from initial conditions $x_0 = [-3,4,1]'$ and the observer states (dashed lines) initialized at $\hat x_0=[0,0,0]'$. Also shown are the states of the sampled-data observer (circles) proposed in Section \ref{sec:SD} with a sampling interval $h=0.1s$. It can be seen that the sampled-data observer is a good approximation to the CT observer, and converges rapidly to the true states.

In Fig.  \ref{fig:ds_30} we show corresponding results with a longer sampling interval of $h=0.3s$. It can now be seen that the behaviour of the CT and sampled-data observers are quite different in the initial transient, nevertheless both converge to the true states.

\begin{figure}\centering
  \includegraphics[width=0.85\columnwidth]{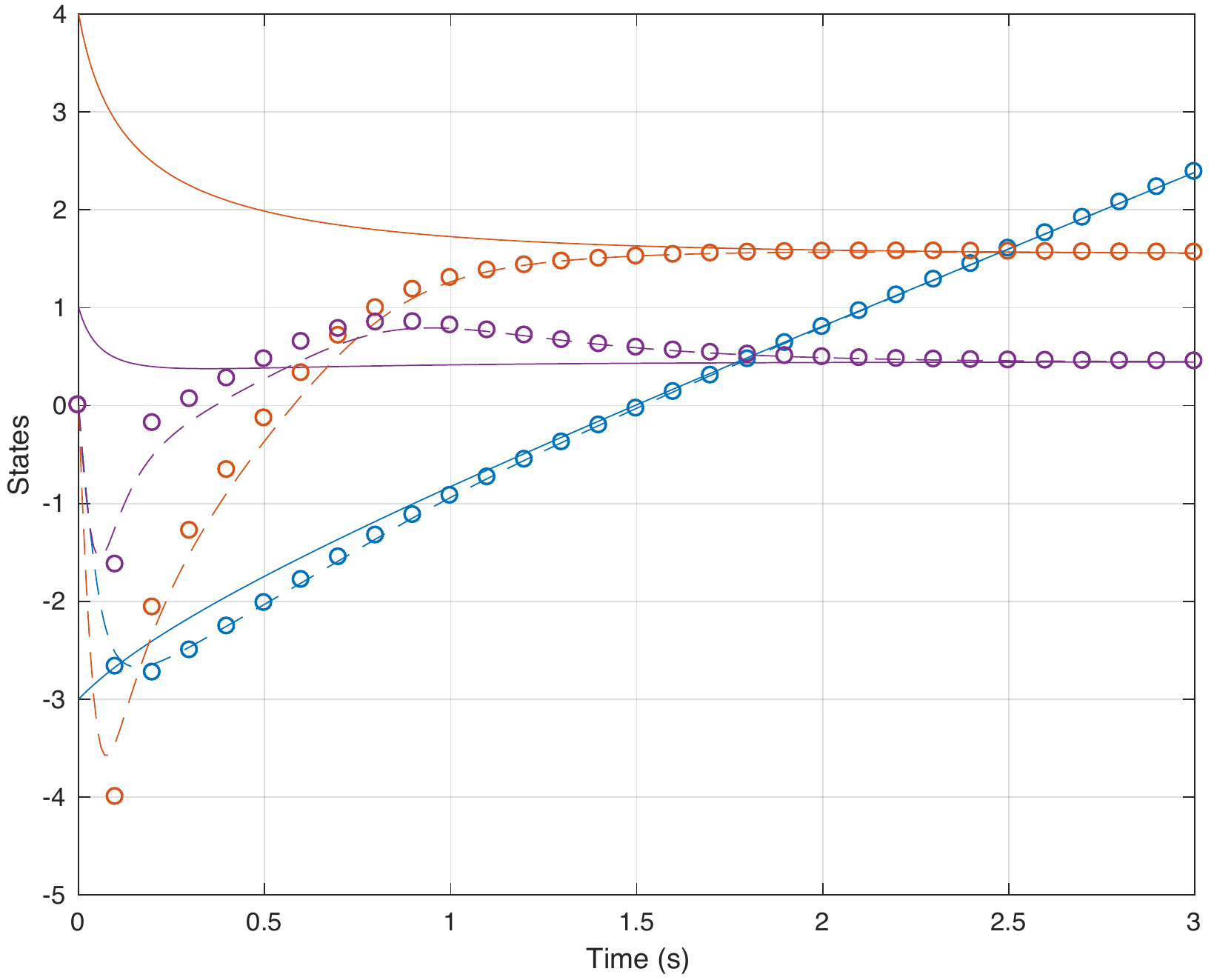}
  \caption{Comparison of true states (solid lines), CT observer (dashed lines), and sampled-data observer (circles) with sampling interval $h=0.1s$.}
  \label{fig:ds_10}
\end{figure}

We next tested the ability of the learning approach in Section \ref{sec:learning} to find an observer with good performance on noisy data. To generate training data we simulated the CT system over 5 seconds, sampled the state $\tilde x$ and output with sampling interval $h=0.1s$, added Gaussian white noise of variance 0.01 to the output to generate $\tilde y$, and then minimized $\hat J^0$ over the set CT1 with $H=10^{-8}I$.

We tested the resulting observer against the previously-designed observer (optimal with respect to $\lambda$) on 50 realisations of validation data, generated in the same way. A boxplot of the resulting mean-square errors is shown in Fig. \ref{fig:lr}. For this example, the learned observer has around half the median error and one third the worst-case error of the observer optimized for $\lambda$.

\begin{figure}\centering
  \includegraphics[width=0.85\columnwidth]{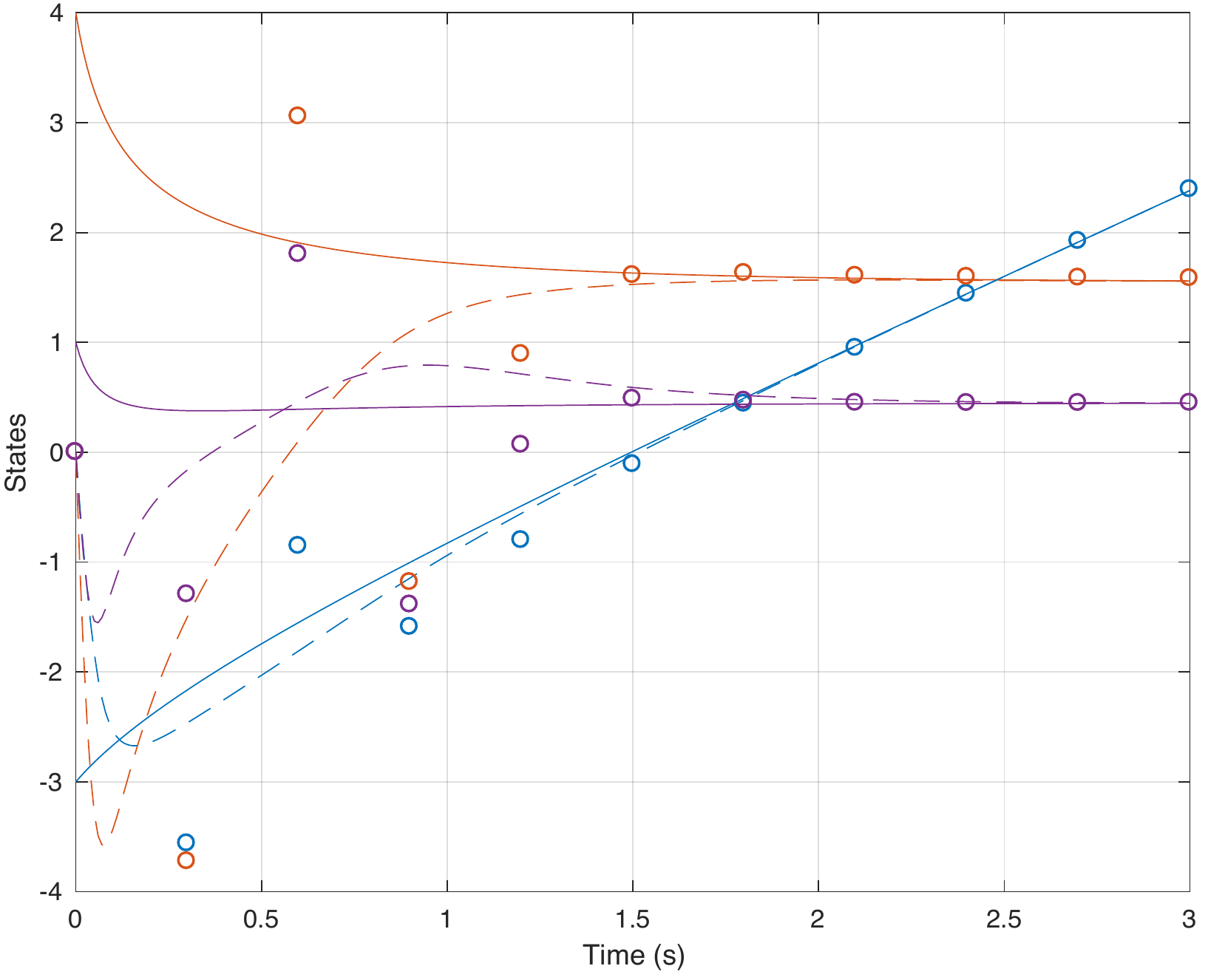}
  \caption{Comparison of true states (solid lines), CT observer (dashed lines), and sampled-data observer (circles) with sampling interval $h=0.3s$.}
  \label{fig:ds_30}
\end{figure}

\begin{figure}\centering
  \includegraphics[width=0.8\columnwidth]{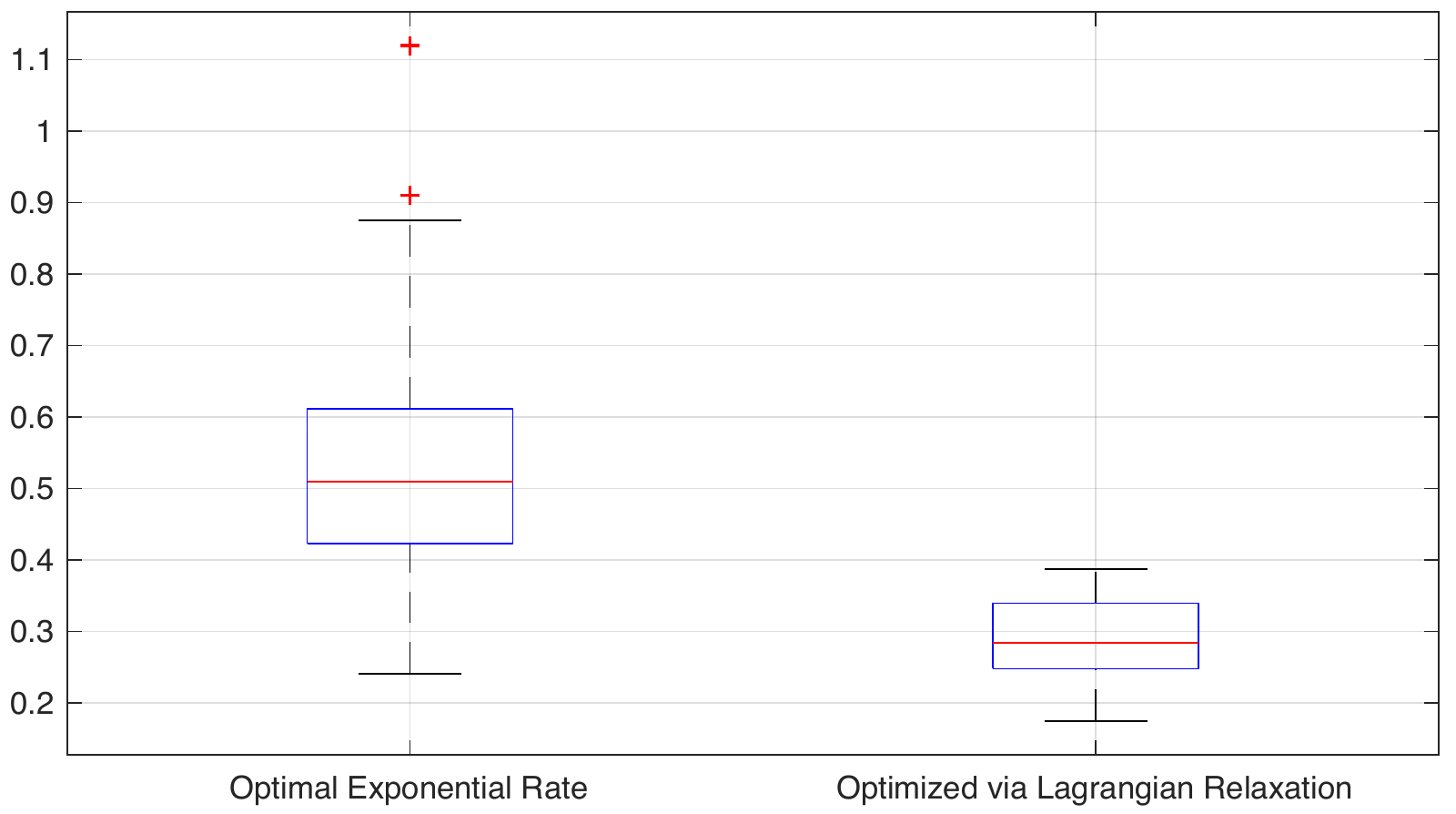}
  \caption{State estimation error on noisy validation data, comparing the observer which is optimal in terms of exponential rate to the one which is optimized for noise performance via Lagrangian relaxation.}
  \label{fig:lr}
\end{figure}

To benchmark the proposed learning method against a ``gold standard'', we now turn to observer design for linear discrete-time systems and compare to the Kalman filter. We compare performance on randomly sampled systems, using different amounts of training data, as follows:
\begin{enumerate}
  \item For each data length $T\in\{100, 250, 500, 1000, 2000\}$ do the following
  \item Sample 20 random 4th-order systems using Matlab {\tt drss} command, eliminating systems with pure integrators. For each system compute the steady-state Kalman filter and do the following:
  \item Generate a random training data set of length $T$ and train a DT observer using the method of Sec. \ref{sec:learning}.
  \item Sample 25 random validation data sets of length and compute the relative error for each, defined as a percentage:
  \[
J_{rel} = 100 \tfrac{\sum_{t=0}^T|\hat x^o_t-\tilde x_t|^2-\sum_{t=0}^T|\hat x^o_t-\tilde x_t|^2}{\sum_{t=0}^T|\hat x^k_t-\tilde x_t|^2}
\]
where $\hat x^o$ is the state estimate from our learned observer, and $\hat x^k$ is the state estimate from the Kalman filter.
\end{enumerate}

Boxplots of the computed $J_{rel}$ for all realisations are shown in Fig. \ref{fig:kalman}, and the median values of $J_{rel}$ for each training-data length are collected in Table \ref{tab:kalman}. It can be seen that with just a few hundred training samples, the performance of the learned estimator comes within around 2\% of the Kalman filter, and on some data sets outperforms the Kalman filter, indicated by negative quantities in Fig \ref{fig:kalman}.

\begin{figure}\centering
  \includegraphics[width=0.85\columnwidth]{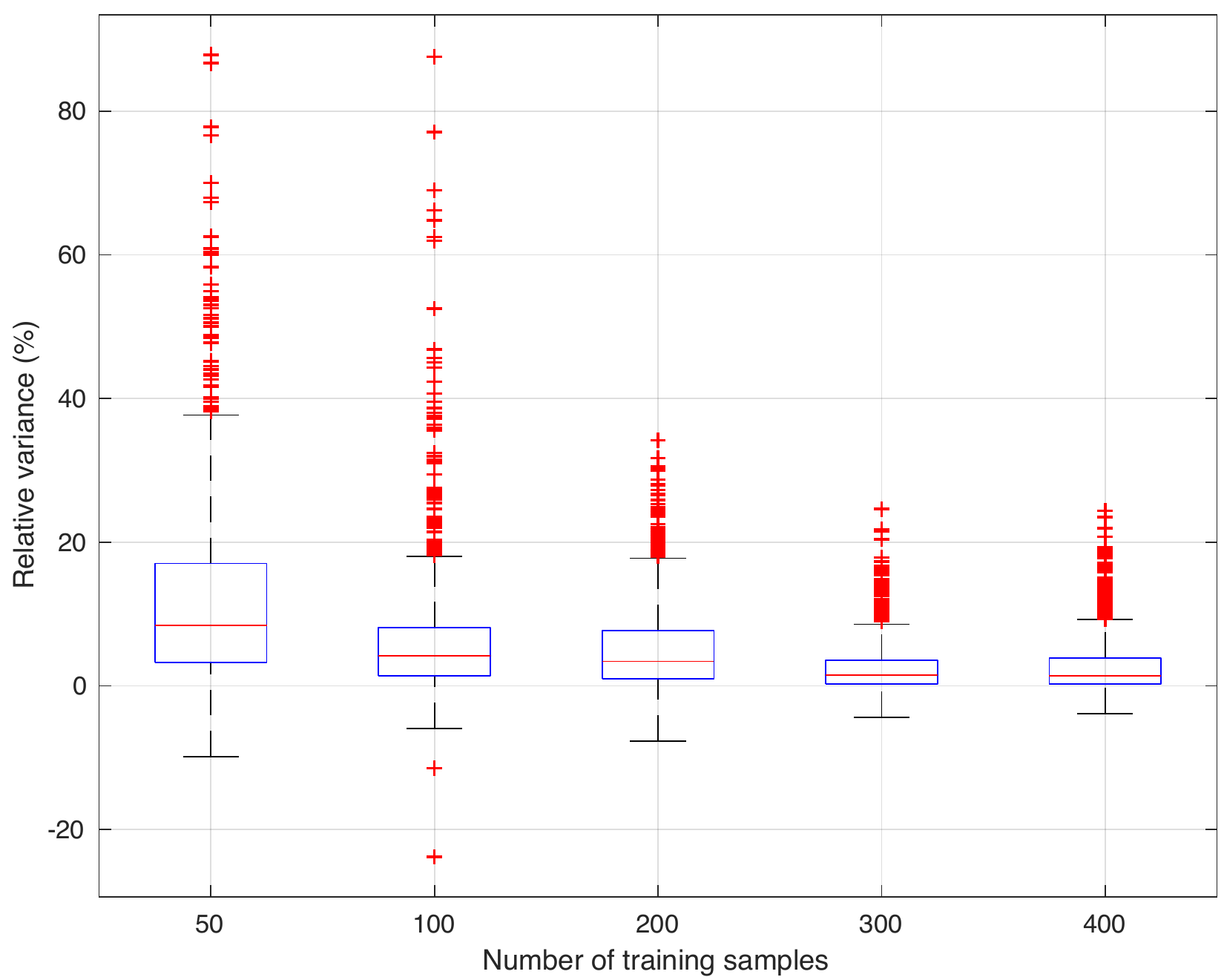}
  \caption{Comparison to the Kalman filter: boxplot of $J_{rel}$ vs number of samples used for training}
  \label{fig:kalman}
\end{figure}

\begin{table}
\caption{Medians of $J_{rel}$ for each training-data length}\label{tab:kalman}
\begin{tabular}{|l|c|c|c|c|c|}
\hline
Samples (N) & 50 & 100 &   200 &   300 &   400 \\
\hline 
Median $J_{rel}$ (\%) & 8.4336  &  4.2223 &   3.4214  &  1.4955  &  1.3731\\ \hline
\end{tabular}
\end{table}

%Median $J_{rel}$ & 1.31 &   0.55 &   0.98 &   0.32  &  0.25\\ \hline
%2nd&    1.704 &   1.293  &  0.577  &  0.359  &  0.299

% Boxplot5: 11.4995    4.3886    2.4580    1.6629
% Boxplot6: 

\section{Conclusions}
In this paper, we have introduced methods for constructing convex sets of contracting CT, DT, and sampled-data observers for a given nonlinear system. We propose optimizing over these sets to based on sampled simulation data to ``learn'' and observer with good performance on noisy data.

While we have focused on the case that the observer state is $\hat x$, the proposed methods can be adapted to reduced-order and excessive-order observers. Also important is the case of ``reduced-complexity'' nonlinear observers, for which the {\em correctness} condition is infeasible, i.e. when the true system dynamics are not in the span of \eqref{eq:linparam}, one can relax  correctness to bounded error in \eqref{eq:correct} subsets of $\R^n\times\R^m$. Investigation of these approaches is underway.

\bibliographystyle{IEEEtran}
\bibliography{Observer.bib}

\begin{thebibliography}{10}
\providecommand{\url}[1]{#1}
\csname url@samestyle\endcsname
\providecommand{\newblock}{\relax}
\providecommand{\bibinfo}[2]{#2}
\providecommand{\BIBentrySTDinterwordspacing}{\spaceskip=0pt\relax}
\providecommand{\BIBentryALTinterwordstretchfactor}{4}
\providecommand{\BIBentryALTinterwordspacing}{\spaceskip=\fontdimen2\font plus
\BIBentryALTinterwordstretchfactor\fontdimen3\font minus
  \fontdimen4\font\relax}
\providecommand{\BIBforeignlanguage}[2]{{%
\expandafter\ifx\csname l@#1\endcsname\relax
\typeout{** WARNING: IEEEtran.bst: No hyphenation pattern has been}%
\typeout{** loaded for the language `#1'. Using the pattern for}%
\typeout{** the default language instead.}%
\else
\language=\csname l@#1\endcsname
\fi
#2}}
\providecommand{\BIBdecl}{\relax}
\BIBdecl

\bibitem{anderson2012optimal}
B.~D.~O. Anderson and J.~B. Moore, \emph{Optimal Filtering}.\hskip 1em plus
  0.5em minus 0.4em\relax {Prentice Hall}, 1979.

\bibitem{petersen1999robust}
I.~R. Petersen and A.~V. Savkin, \emph{Robust {{Kalman}} Filtering for Signals
  and Systems with Large Uncertainties}.\hskip 1em plus 0.5em minus 0.4em\relax
  {Springer Science \& Business Media}, 1999.

\bibitem{doucet2001sequential}
A.~Doucet, N.~De~Freitas, and N.~Gordon, \emph{Sequential Monte Carlo Methods
  in Practice}.\hskip 1em plus 0.5em minus 0.4em\relax {Springer Verlag}, 2001.

\bibitem{rao2003constrained}
C.~V. Rao, J.~B. Rawlings, and D.~Q. Mayne, ``Constrained state estimation for
  nonlinear discrete-time systems: {{Stability}} and moving horizon
  approximations,'' \emph{IEEE transactions on automatic control}, vol.~48,
  no.~2, pp. 246--258, 2003.

\bibitem{shamma_approximate_1997}
J.~S. Shamma and K.-Y. Tu, ``Approximate set-valued observers for nonlinear
  systems,'' \emph{IEEE Transactions on Automatic Control}, vol.~42, no.~5, pp.
  648--658, May 1997.

\bibitem{praly2015observers}
L.~Praly, ``Observers for {{Nonlinear Systems}},'' \emph{Encyclopedia of
  Systems and Control}, pp. 935--943, 2015.

\bibitem{krener_linearization_1983}
A.~J. Krener and A.~Isidori, ``Linearization by output injection and nonlinear
  observers,'' \emph{Systems \& Control Letters}, vol.~3, no.~1, pp. 47--52,
  Jun. 1983.

\bibitem{khalil_high-gain_2014}
H.~K. Khalil and L.~Praly, ``\BIBforeignlanguage{en}{High-gain observers in
  nonlinear feedback control},'' \emph{\BIBforeignlanguage{en}{International
  Journal of Robust and Nonlinear Control}}, vol.~24, no.~6, pp. 993--1015,
  Apr. 2014.

\bibitem{karagiannis_invariant_2008}
D.~Karagiannis, D.~Carnevale, and A.~Astolfi, ``Invariant {{Manifold Based
  Reduced}}-{{Order Observer Design}} for {{Nonlinear Systems}},'' \emph{IEEE
  Transactions on Automatic Control}, vol.~53, no.~11, pp. 2602--2614, Dec.
  2008.

\bibitem{andrieu_existence_2006}
V.~Andrieu and L.~Praly, ``On the {{Existence}} of a
  {{Kazantzis}}--{{Kravaris}}/{{Luenberger Observer}},'' \emph{SIAM Journal on
  Control and Optimization}, vol.~45, no.~2, pp. 432--456, Jan. 2006.

\bibitem{arcak_nonlinear_2001}
M.~Arcak and P.~Kokotovi{\'c}, ``Nonlinear observers: A circle criterion design
  and robustness analysis,'' \emph{Automatica}, vol.~37, no.~12, pp.
  1923--1930, Dec. 2001.

\bibitem{fan_observer_2003}
X.~Fan and M.~Arcak, ``\BIBforeignlanguage{en}{Observer design for systems with
  multivariable monotone nonlinearities},''
  \emph{\BIBforeignlanguage{en}{Systems \& Control Letters}}, vol.~50, no.~4,
  pp. 319--330, Nov. 2003.

\bibitem{coutinho_robust_2009}
D.~Coutinho, C.~{de Souza}, K.~Barbosa, and A.~Trofino, ``Robust {{Linear}}
  \${{H}}\_\{$\backslash$infty\}\$ {{Filter Design}} for a {{Class}} of
  {{Uncertain Nonlinear Systems}}: {{An LMI Approach}},'' \emph{SIAM Journal on
  Control and Optimization}, vol.~48, no.~3, pp. 1452--1472, Jan. 2009.

\bibitem{ebenbauer_polynomial_2005}
C.~Ebenbauer, J.~Renz, and F.~Allgower, ``Polynomial {{Feedback}} and
  {{Observer Design}} using {{Nonquadratic Lyapunov Functions}},'' in
  \emph{Proceedings of the 44th {{IEEE Conference}} on {{Decision}} and
  {{Control}}}, Dec. 2005, pp. 7587--7592.

\bibitem{parrilo2003semidefinite}
P.~A. Parrilo, ``Semidefinite programming relaxations for semialgebraic
  problems,'' \emph{Mathematical Programming}, vol.~96, no.~2, pp. 293--320,
  2003.

\bibitem{Lohmiller98}
W.~Lohmiller and J.-J.~E. Slotine, ``On {{Contraction Analysis}} for
  {{Non}}-linear {{Systems}},'' \emph{Automatica}, vol.~34, no.~6, pp.
  683--696, Jun. 1998.

\bibitem{sanfelice_convergence_2012}
R.~Sanfelice and L.~Praly, ``Convergence of {{Nonlinear Observers}} on {{With}}
  a {{Riemannian Metric}} ({{Part I}}),'' \emph{IEEE Transactions on Automatic
  Control}, vol.~57, no.~7, pp. 1709--1722, Jul. 2012, 00014.

\bibitem{dani_observer_2015}
A.~P. Dani, S.~J. Chung, and S.~Hutchinson, ``Observer {{Design}} for
  {{Stochastic Nonlinear Systems}} via {{Contraction}}-{{Based Incremental
  Stability}},'' \emph{IEEE Transactions on Automatic Control}, vol.~60, no.~3,
  pp. 700--714, Mar. 2015.

\bibitem{lewis1949metric}
D.~Lewis, ``Metric properties of differential equations,'' \emph{American
  Journal of Mathematics}, pp. 294--312, 1949.

\bibitem{forni2014differential}
F.~Forni and R.~Sepulchre, ``A {{Differential Lyapunov Framework}} for
  {{Contraction Analysis}},'' \emph{Automatic Control, IEEE Transactions on},
  vol.~59, no.~3, pp. 614--628, 2014.

\bibitem{manchester2013transverse}
I.~R. Manchester and J.-J.~E. Slotine, ``Transverse {{Contraction}} criteria
  for existence, stability, and robustness of a limit cycle,'' \emph{Systems \&
  Control Letters}, vol.~62, pp. 32--38, 2013.

\bibitem{manchester_control_2017}
I.~R. Manchester and J.~J.~E. Slotine, ``Control {{Contraction Metrics}}:
  {{Convex}} and {{Intrinsic Criteria}} for {{Nonlinear Feedback Design}},''
  \emph{IEEE Transactions on Automatic Control}, vol.~62, no.~6, pp.
  3046--3053, Jun. 2017.

\bibitem{arcak_framework_2004}
M.~Arcak and D.~Ne{\v s}i{\'c}, ``A framework for nonlinear sampled-data
  observer design via approximate discrete-time models and emulation,''
  \emph{Automatica}, vol.~40, no.~11, pp. 1931--1938, Nov. 2004.

\bibitem{Sjoberg95}
J.~Sj{\"o}berg, Q.~Zhang, L.~Ljung, A.~Benveniste, B.~Delyon, P.-Y. Glorennec,
  H.~Hjalmarsson, and A.~Juditsky, ``Nonlinear black-box modeling in system
  identification: A unified overview,'' \emph{Automatica}, vol.~31, no.~12, pp.
  1691--1724, 1995.

\bibitem{LjungBook}
L.~Ljung, \emph{System {{Identification}}: {{Theory}} for the {{User}}},
  3rd~ed.\hskip 1em plus 0.5em minus 0.4em\relax Englewood Cliffs, New Jersey,
  USA: {Prentice Hall}, 1999.

\bibitem{pascanu2013difficulty}
R.~Pascanu, T.~Mikolov, and Y.~Bengio, ``On the difficulty of training
  recurrent neural networks,'' in \emph{International {{Conference}} on
  {{Machine Learning}}}, 2013, pp. 1310--1318.

\bibitem{tobenkin2010convex}
M.~M. Tobenkin, I.~R. Manchester, J.~Wang, A.~Megretski, and R.~Tedrake,
  ``Convex optimization in identification of stable non-linear state space
  models,'' in \emph{49th {{IEEE Conference}} on {{Decision}} and {{Control}}
  ({{CDC}})}.\hskip 1em plus 0.5em minus 0.4em\relax {IEEE}, 2010.

\bibitem{tobenkin_convex_2017}
M.~M. Tobenkin, I.~R. Manchester, and A.~Megretski, ``Convex
  {{Parameterizations}} and {{Fidelity Bounds}} for {{Nonlinear
  Identification}} and {{Reduced}}-{{Order Modelling}},'' \emph{IEEE
  Transactions on Automatic Control}, vol.~62, no.~7, pp. 3679--3686, Jul.
  2017.

\bibitem{wang_partial_2004}
W.~Wang and J.-J.~E. Slotine, ``\BIBforeignlanguage{en}{On partial contraction
  analysis for coupled nonlinear oscillators},''
  \emph{\BIBforeignlanguage{en}{Biological Cybernetics}}, vol.~92, no.~1, pp.
  38--53, Dec. 2004.

\bibitem{dahlquist_special_1963}
G.~G. Dahlquist, ``\BIBforeignlanguage{en}{A special stability problem for
  linear multistep methods},'' \emph{\BIBforeignlanguage{en}{BIT Numerical
  Mathematics}}, vol.~3, no.~1, pp. 27--43, Mar. 1963.

\bibitem{Umenberger16}
J.~Umenberger and I.~R. Manchester, ``Specialized {{Algorithm}} for
  {{Identification}} of {{Stable Linear Systems}} using {{Lagrangian
  Relaxation}},'' in \emph{Proc. 2016 {{American Control Conference}}}, Boston,
  MA, 2016.

\bibitem{andersen_implementation_2010}
M.~S. Andersen, J.~Dahl, and L.~Vandenberghe,
  ``\BIBforeignlanguage{en}{Implementation of nonsymmetric interior-point
  methods for linear optimization over sparse matrix cones},''
  \emph{\BIBforeignlanguage{en}{Mathematical Programming Computation}}, vol.~2,
  no. 3-4, pp. 167--201, Dec. 2010.

\bibitem{YALMIP}
J.~Lofberg, ``Yalmip : A toolbox for modeling and optimization in {MATLAB},''
  in \emph{Proceedings of the CACSD Conference}, Taipei, Taiwan, 2004.

\end{thebibliography}

\end{document}